\documentclass[letterpaper, 10 pt, conference]{ieeeconf}
\IEEEoverridecommandlockouts 
\usepackage{graphics}
\usepackage{epsfig}
\usepackage{amsmath}
\usepackage{amssymb}
\usepackage{xcolor}
\usepackage{multirow}

\usepackage{amsthm}
\theoremstyle{definition}
\newtheorem{lemma}{\normalfont \bfseries Lemma}

\newtheorem{theorem}{\normalfont\bfseries Theorem}
\newtheorem{assumption}{\normalfont\bfseries Assumption}

\newtheorem{definition}{\normalfont\bfseries Definition}

\newtheorem{remark}{\normalfont\bfseries Remark}
\newtheorem{example}{\normalfont\bfseries Example}

\newcommand{\R}{\mathbb{R}}
\newcommand{\Kinf}{\mathcal{K}_\infty}
\newcommand{\U}{\mathcal{U}}
\renewcommand{\S}{\mathcal{S}}
\newcommand{\Sb}{\mathcal{S}_{\rm b}}
\newcommand{\SI}{\mathcal{S}_{\rm I}}
\newcommand{\Sd}{\mathcal{S}_{\rm d}}
\newcommand{\C}{\mathcal{C}}
\newcommand{\CI}{\mathcal{C}_{\rm I}}
\newcommand{\Cd}{\mathcal{C}_{\rm d}}
\newcommand{\kd}{k_{\rm d}}
\newcommand{\hb}{h_{\rm b}}
\newcommand{\kb}{k_{\rm b}}
\newcommand{\alphab}{\alpha_{\rm b}}
\newcommand{\gammab}{\gamma_{\rm b}}
\newcommand{\sigmab}{\sigma_{\rm b}}
\newcommand{\phib}{\phi_{\rm b}}

\newcommand{\Hb}{H_{\rm b}}

\newcommand{\gradh}{\nabla h}
\newcommand{\gradhb}{\nabla \hb}
\newcommand{\gradP}{\nabla P}

\newcommand{\derp}[2]{\frac{\partial #1}{\partial #2}}

\title{\LARGE \bf
Safety-Critical Control with Bounded Inputs via Reduced Order Models
}

\author{Tamas G. Molnar and Aaron D. Ames%
\thanks{*This research is supported in part by the National Science Foundation, CPS Award \#1932091, Dow (\#227027AT) and Aerovironment.}%
\thanks{The Authors are with the Department of Mechanical and Civil Engineering, California Institute of Technology, Pasadena, CA 91125, USA.
{\tt\small \{tmolnar, ames\}@caltech.edu}}%
}

\begin{document}

\maketitle
\thispagestyle{empty}
\pagestyle{empty}

\begin{abstract}
Guaranteeing safe behavior on complex autonomous systems---from cars to walking robots---is challenging due to the inherently high dimensional nature of these systems and the corresponding complex models that may be difficult to determine in practice.  
With this as motivation, this paper presents a safety-critical control framework that leverages reduced order models to ensure safety on the full order dynamics---even when these models are subject to disturbances and bounded inputs (e.g., actuation limits).  
To handle input constraints, the backup set method is reformulated in the context of reduced order models, and conditions for the provably safe behavior of the full order system are derived.
Then, the input-to-state safe backup set method is introduced to provide robustness against discrepancies between the reduced order model and the actual system.
Finally, the proposed framework is demonstrated in high-fidelity simulation, where a quadrupedal robot is safely navigated around an obstacle with legged locomotion by the help of the unicycle model.

\end{abstract}

\section{INTRODUCTION}
\label{sec:intro}




Real-life engineering systems often exhibit complicated, nonlinear and high-dimensional dynamic behavior.  This is especially true of autonomous (robotic) systems, where dynamics play a key role in achieving desired behaviors.  
This makes them challenging to control, and to attain formal guarantees of stable or safe evolution for the closed control loop.
To tackle such complex control problems, simplified, reduced order models (ROMs) of the dynamic behavior are often utilized during controller synthesis with great practical success \cite{fawcett2022toward, xiong2022}.  Yet there is often a theoretic gap between behaviors certifiable on the ROM and the resulting behaviors observed on the full order system (FOS).

In this paper, we focus on the role of ROMs in safety-critical control.
Given an accurate dynamical model, there exist tools to synthesize controllers that provide formal guarantees of safety.
For example, {\em control barrier functions (CBFs)}~\cite{AmesXuGriTab2017} have been proposed to achieve this goal, and they have been proven to be successful in a wide variety of applications from multi-robot systems~\cite{glotfelter2017nonsmooth} to spacecraft docking~\cite{dunlap2022comparing}.
In many applications, a significant challenge is maintaining safety with limited actuation:
most physical systems have finite actuation capability, which manifests itself in the underlying models as constraints on the control input.
Several methods have been proposed for input constrained safety-critical control, including the backup set method~\cite{gurriet2020scalable}, input constrained CBFs~\cite{agrawal2021safe}, and neural CBFs~\cite{liu2022safe}.
While these approaches have shown success in various domains, a general approach remains elusive.

\begin{figure}
\centering
\includegraphics[scale=1]{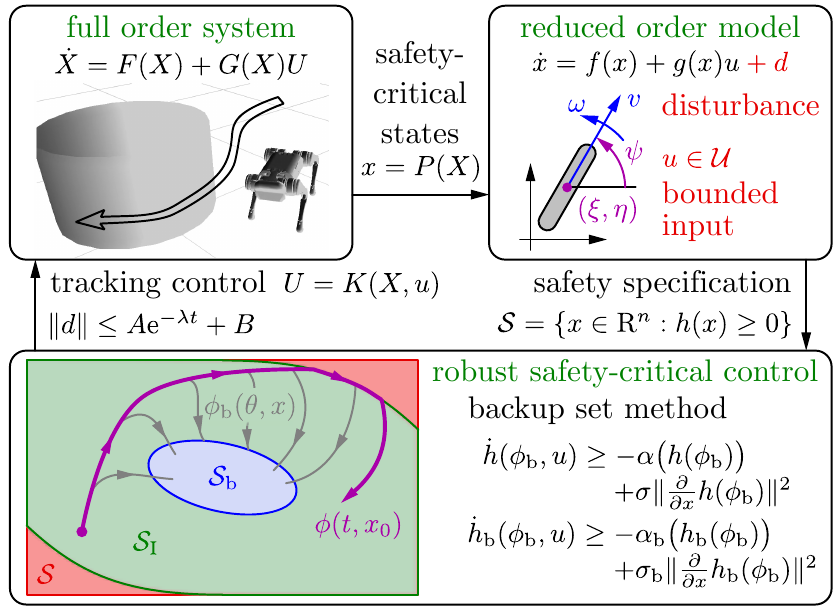}
\vspace{-0.1cm}
\caption{
Overview of the proposed safety-critical control framework.
}
\label{fig:concept}
\vspace{-0.4cm}
\end{figure}

Another important challenge in safety-critical control is that no ROM is ever fully accurate: there is always discrepancy between the ROM and the actual FOS.
Consequently, robustness is of key importance, and one needs to ensure that safety is preserved even under such discrepancies---and with limited actuation.
There exist CBF formulations that provide robustness against disturbances, such as the approaches of input-to-state safety~\cite{ames2019issf, Alan2022} and robust CBFs~\cite{jankovic2018robust}.
These formulations, however, have not yet accommodated input constraints.
Meanwhile, the above approaches that address input constraints have not yet been endowed with robustness.
On the other hand, there exist robust reachability approaches that handle both input constraints and disturbances~\cite{ding2011reachability, bansal2017HJreachability, Kousik2020, Choi2021}, but these methods typically suffer from the curse of dimensionality and become intractable for higher dimensional ROMs.

This paper presents a robust safety-critical control framework, illustrated in Fig.~\ref{fig:concept}, wherein input constrained ROMs and CBFs are leveraged to achieve formal safety guarantees on systems with complex full order dynamics. 
To this end, we make the following three key contributions.
First, the backup set method is reformulated in the context of ROMs, and conditions for provably safe behavior are given that account for the discrepancy between the ROM and the FOS that tracks it.
Second, the {\em input-to-state safe backup set method} is introduced to provide robustness against the discrepancy with less restrictive conditions.
Third, the method is implemented in the context of an obstacle avoidance problem, wherein safe walking on a quadrupedal robot using the unicycle ROM is demonstrated
in high-fidelity simulation.






\section{SAFETY UNDER INPUT CONSTRAINTS}
\label{sec:backup}

Consider the control-affine system:
\begin{equation}
    \dot{x} = f(x) + g(x) u,
\label{eq:system}
\end{equation}
with state ${x \in \R^n}$, input ${u \in \U}$, convex admissible input set ${\U \subseteq \R^m}$, and locally Lipschitz continuous functions ${f : \R^n \to \R^n}$, ${g : \R^n \to \R^{n \times m}}$.
Consider a controller ${k : \R^n \to \U}$, ${u = k(x)}$, that yields the closed control loop:
\begin{equation}
    \dot{x} = f(x) + g(x) k(x),
\label{eq:closedloop}
\end{equation}
associated with the initial condition ${x(0) = x_0 \in \R^n}$.
If $k$ is locally Lipschitz continuous, the closed-loop system has a unique solution ${\phi(t,x_0)}$ over an interval of existence.
For simplicity, we assume that the solution exists for all ${t \geq 0}$.

Our goal is to design the controller $k$ such that the closed-loop system is safe.
Specifically, we consider the system to be safe if its state $x$ is located within a safe set ${\S \subset \R^n}$.
For the safe evolution of the closed control loop, we require the forward invariance of the safe set $\S$ along~(\ref{eq:closedloop}).
\begin{definition}
Given ${k: \R^n \to \U}$, set ${\S \subset \R^n}$ is {\em forward invariant} along~(\ref{eq:closedloop}) if ${x_0 \in \S \implies \phi(t,x_0) \in \S}$, ${\forall t \geq 0}$.
\end{definition}
\noindent This requirement can be met only if $\S$ is control invariant.
\begin{definition}
Set $\S \subset \R^n$ is {\em control invariant} if there exists ${k: \R^n \to \U}$ such that $\S$ is forward invariant along~(\ref{eq:closedloop}).
\end{definition}

\subsection{Control Barrier Functions}

Control barrier functions~\cite{AmesXuGriTab2017} provide a powerful tool for safe control design, hence we briefly revisit this method.
Throughout the paper, we consider safe sets defined as the 0-superlevel set of a function ${h : \R^n \to \R}$:
\begin{equation}
    \S = \{x \in \R^n: h(x) \geq 0 \},
\label{eq:safeset}
\end{equation}
such that $h$ is continuously differentiable and zero is a regular value of $h$, i.e., ${h(x) = 0 \implies \gradh(x) \neq 0}$.

\begin{definition}\label{def:CBF}
Function $h$ is a {\em control barrier function (CBF)} for~(\ref{eq:system}) on $\S$ if there exists ${\alpha \in \Kinf}$ such that\footnote{Function ${\alpha : \R_{\geq 0} \to \R_{\geq 0}}$ is of class-$\Kinf$ (${\alpha \in \Kinf}$) if it is continuous, ${\alpha(0)=0}$ and ${\lim_{r \to \infty} \alpha(r) = \infty}$. Note that extended class-$\Kinf$ functions defined over $\R$ are also used to ensure the attractivity of the safe set.}:
\begin{equation}
    \sup_{u \in \U} \dot{h}(x,u) > - \alpha \big( h(x) \big)
\label{eq:CBF_condition}
\end{equation}
holds ${\forall x \in \S}$, where:
\begin{equation}
    \dot{h}(x,u) = \gradh(x) (f(x) + g(x) u).
\end{equation}
\end{definition}


Given a CBF, \cite{AmesXuGriTab2017} established the following safety result.
\begin{theorem}[\cite{AmesXuGriTab2017}] \label{thm:CBF}
\textit{
If $h$ is a CBF for~(\ref{eq:system}) on $\S$, then any locally Lipschitz continuous controller ${k : \R^n \to \U}$
satisfying:
\begin{equation}
    \dot{h} \big( x, k(x) \big) \geq - \alpha \big( h(x) \big)
\label{eq:safety_condition}
\end{equation}
${\forall x \in \S}$ renders $\S$ forward invariant along~(\ref{eq:closedloop}).}
\end{theorem}

Condition~(\ref{eq:safety_condition})
can be used as constraint when synthesizing safe controllers.
For example, given a desired controller ${\kd : \R^n \to \U}$, the following quadratic program-based controller can be used for safety-critical control: 
\begin{align}
\begin{split}
    k(x) = \underset{u \in \U}{\operatorname{argmin}} & \quad \| u - \kd(x) \|_{\Gamma}^2 \\
    \text{s.t.} & \quad \dot{h}(x,u) \geq - \alpha \big( h(x) \big),
\end{split}
\label{eq:QP}
\end{align}
where ${\| u \|_{\Gamma}^2 = u^\top \Gamma u}$ and ${\Gamma \in \R^{m \times m}}$ is a positive definit weight matrix that can be tuned.
 
\subsection{Backup Set Method}

While CBFs provide safe behavior, it is nontrivial to verify that a certain choice of $h$ is indeed a CBF satisfying~(\ref{eq:CBF_condition}), especially with bounded inputs (${\U \subset \R^m}$).
An arbitrary $h$ may not have control invariant 0-superlevel set, it may not be a CBF, and
safe inputs satisfying~(\ref{eq:safety_condition}) may not exist.
Consequently, optimization problems like~(\ref{eq:QP}) may be infeasible with input bounds.
The backup set method~\cite{gurriet2020scalable} was proposed to solve this problem, by synthesizing control invariant sets and corresponding safe controllers via the CBF framework.

The backup set method is described as follows; while examples are given below and in~\cite{gurriet2020scalable,chen2021backup}.
First, one must specify a control invariant subset of $\S$, called the {\em backup set}:
\begin{equation}
    \Sb = \{x \in \R^n: \hb(x) \geq 0 \},
\label{eq:backupset}
\end{equation}
such that ${\hb : \R^n \to \R}$ is continuously differentiable, zero is a regular value of $\hb$, i.e., ${\hb(x) = 0 \implies \gradhb(x) \neq 0}$, and  ${\Sb \subseteq \S}$.
Furthermore, one must define a {\em backup controller} ${\kb : \R^n \to \U}$ that renders the backup set forward invariant along the closed-loop system:
\begin{equation}
    \dot{x} = f(x) + g(x) \kb(x) \triangleq f_{\rm b}(x).
\label{eq:backupsystem}
\end{equation}
We denote the solution of~(\ref{eq:backupsystem}) with ${x(0) = x_0 \in \R^n}$ by $\phib(t,x_0)$.
To summarize, the choice of backup set and backup controller must satisfy the following assumption.

\begin{assumption} \label{assum:backup}
The backup set ${\Sb \subseteq \S}$ is control invariant, and the backup controller $\kb$ renders $\Sb$ forward invariant along~(\ref{eq:backupsystem}) while satisfying the input constraints:
\begin{equation}
    x \in \Sb \implies \phib(\theta,x) \in \Sb \subseteq \S, \ \forall \theta \geq 0,
\end{equation}
and ${\kb(x) \in \U}$, ${\forall x \in \S}$.
\end{assumption}

Finding a control invariant subset $\Sb$ is considerably less difficult than verifying that a given $\S$ is control invariant.
With this, by construction, we have a control invariant set $\Sb$ and a safe controller $\kb$ at our disposal.
However, methods for constructing $\Sb$ (see examples in~\cite{gurriet2020scalable,chen2021backup}) may result in a very small set, hence operating the system directly within $\Sb$ may make the behavior overly conservative.
To reduce this conservatism, we enlarge $\Sb$ to the set ${\SI \subseteq \S}$:
\begin{equation}
    \SI = \left\{ x \in \R^n :
    \begin{array}{l}
    \phib(\theta,x) \in \S, \ \forall \theta \in [0,T], \\
    \phib(T,x) \in \Sb 
    \end{array}
    \right\},
\label{eq:invariantset}
\end{equation}
with ${T \geq 0}$; cf.~Fig.~\ref{fig:concept}.
Note that $T$ is a design parameter, the size of $\SI$ increases with $T$, and ${T=0}$ yields ${\SI = \Sb}$.

\begin{lemma}[\cite{gurriet2020scalable}] \label{lemma:backup_invariance}
\textit{
The set ${\SI}$ is control invariant, and the backup controller $\kb$ renders $\SI$ forward invariant along~(\ref{eq:backupsystem}):
\begin{equation}
    x \in \SI \implies \phib(\vartheta,x) \in \SI, \ \forall \vartheta \geq 0.
\label{eq:SI_invariance}
\end{equation}
}
\end{lemma}
\noindent For the proofs of Lemmas~\ref{lemma:backup_invariance} and~\ref{lemma:backup_feasibility}, see the Appendix.

Thus,~(\ref{eq:invariantset}) yields a larger, practically more useful control invariant set $\SI$ than the backup set $\Sb$; see~\cite{chen2021backup} for an analysis about the size of $\SI$.
We use $\SI$ to provide safety, based on the framework of CBFs.
We rely on
the derivatives:
\begin{align}
\begin{split}
    \dot{h} \big( \phib(\theta,x), u \big) & = \derp{h \big( \phib(\theta, x) \big)}{x} \big( f(x) + g(x) u \big), \\
    \dot{h}_{\rm b} \big( \phib(T,x), u \big) & = \derp{\hb \big( \phib(T, x) \big)}{x} \big( f(x) + g(x) u \big).
\end{split}
\label{eq:hdot}
\end{align}
Then, we can state that the backup controller $\kb$ satisfies safety conditions analogous to~(\ref{eq:safety_condition}).

\begin{lemma}[\cite{gurriet2020scalable}] \label{lemma:backup_feasibility}
\textit{
There exist $\alpha, \alphab \in \Kinf$ such that $\forall x \in \SI$:
\begin{align}
\begin{split}
    \dot{h} \big( \phib(\theta,x), \kb(x) \big) & \!\geq\! - \alpha \big( h(\phib(\theta,x)) \big), \ \forall \theta \!\in\! [0,T], \\
    \dot{h}_{\rm b} \big( \phib(T,x), \kb(x) \big) & \!\geq\! - \alphab \big( \hb(\phib(T,x)) \big).
\end{split}
\label{eq:backup_feasibility}
\end{align}
}
\end{lemma}

This leads to the main result of the backup set method.
\begin{theorem}[\cite{gurriet2020scalable}] \label{thm:backup}
\textit{
Consider system~(\ref{eq:system}), set $\S$ in~(\ref{eq:safeset}), set $\Sb$ in~(\ref{eq:backupset}), Assumption~\ref{assum:backup}, and set $\SI$ in~(\ref{eq:invariantset}).
Then, there exist ${\alpha, \alphab \in \Kinf}$ such that a controller ${k : \R^n \to \U}$ satisfying:
\begin{align}
\begin{split}
    \dot{h} \big( \phib(\theta,x), k(x) \big) & \!\geq\! - \alpha \big( h(\phib(\theta,x)) \big), \ \forall \theta \!\in\! [0,T], \\
    \dot{h}_{\rm b} \big( \phib(T,x), k(x) \big) & \!\geq\! - \alphab \big( \hb(\phib(T,x)) \big).
\end{split}
\label{eq:backup_condition}
\end{align}
${\forall x \in \SI}$ is guaranteed to exist.
Moreover, any locally Lipschitz continuous controller ${k : \R^n \to \U}$
that satisfies~(\ref{eq:backup_condition})
${\forall x \in \SI}$ renders ${\SI \subseteq \S}$ forward invariant along~(\ref{eq:closedloop}).
}
\end{theorem}

\begin{proof}
The existence of a controller $k$ satisfying~(\ref{eq:backup_condition}) follows from Lemma~\ref{lemma:backup_feasibility}, since $\kb$ is such a controller.
The forward invariance of $\SI$ is the consequence of Theorem~\ref{thm:CBF}.
\end{proof}

\subsection{Implementation in Optimization Problems}

Theorem~\ref{thm:backup} can be directly used for controller synthesis, for example, by using~(\ref{eq:backup_condition}) in optimization problems like~(\ref{eq:QP}):
\begin{align}
\begin{split}
    k(x) \!=\! \underset{u \in \U}
    {\operatorname{argmin}}
    & \ \| u - k_{\rm d}(x) \|_{\Gamma}^2 \\
    \text{s.t.}
    & \ \dot{h} \big( \phib(\theta,x), u \big) \!\!\geq\! - \alpha \big( h(\phib(\theta,x)) \big), \forall \theta \!\in\! [0,T], \\
    & \ \dot{h}_{\rm b} \big( \phib(T,x), u \big) \!\!\geq\! - \alphab \big( \hb(\phib(T,x)) \big).
\end{split}
\label{eq:backupOP}
\end{align}
Note that the constraints are affine in $u$, cf.~(\ref{eq:hdot}), hence the optimization problem is convex, and it leads to a quadratic program if ${u \in \U}$ is also described by affine constraints.
Moreover, unlike~(\ref{eq:QP}), the optimization problem~(\ref{eq:backupOP}) is guaranteed to be feasible even if $h$ is not verified to be a CBF.

\begin{figure*}
\centering
\includegraphics[scale=1]{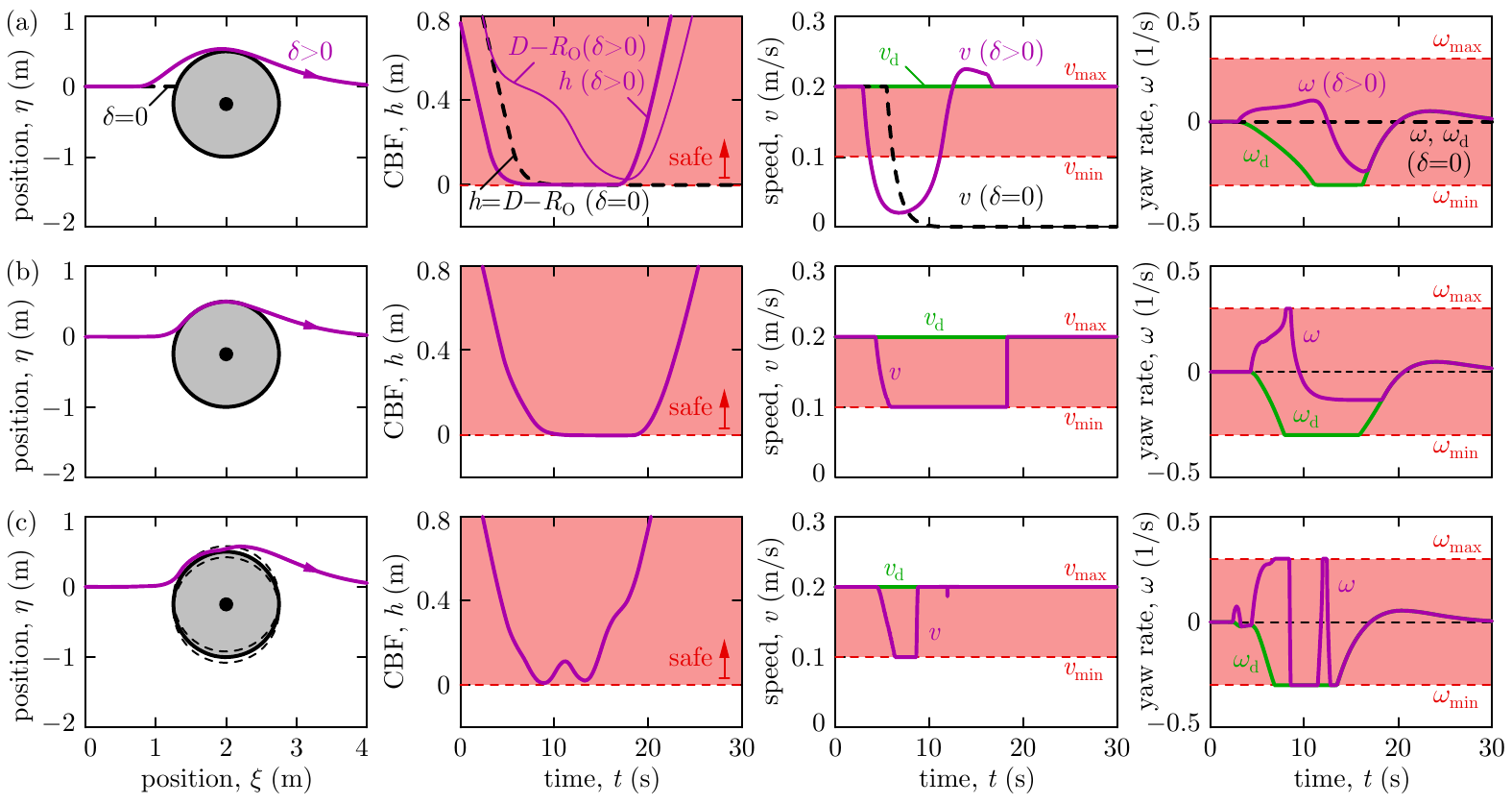}
\caption{
Safety-critical control of the unicycle model for obstacle avoidance.
(a) The CBF-based controller~(\ref{eq:QP}) maintains safety, without limits on the inputs (speed and yaw rate).
(b) The backup set method-based controller~(\ref{eq:backupOPdiscretized}) enforces safety with input constraints.
(c) Controller~(\ref{eq:backupOPdiscretized}) handles moving obstacle.
}
\label{fig:unicycle}
\end{figure*}

The constraints of~(\ref{eq:backupOP}) contain the terms in~(\ref{eq:hdot}), where:
\begin{align}
\begin{split}
    \derp{h \big( \phib(\theta, x) \big)}{x} & = \gradh \big( \phib(\theta, x) \big) \derp{ \phib(\theta, x)}{x}, \\
    \derp{\hb \big( \phib(T, x) \big)}{x} & = \gradhb \big( \phib(T, x) \big) \derp{ \phib(T, x)}{x}.
\end{split}
\end{align}
Here, ${Q(\theta, x) \triangleq \partial \phib(\theta, x) / \partial x}$ is the sensitivity of the flow $\phib(\theta,x)$ to its initial condition $x$.
$\phib(\theta, x)$ and $Q(\theta, x)$ can be calculated together by solving the initial value problem: 
\begin{align}
\begin{split}
    \phib'(\theta, x) & = f_{\rm b} \big( \phib(\theta, x) \big), \qquad \qquad \;\; \phib(0, x) = x, \\
    Q'(\theta, x) & = \derp{f_{\rm b}}{x} \big( \phib(\theta, x) \big) Q(\theta, x), \quad Q(0, x) = I,
\end{split}
\end{align}
where
prime denotes derivative with respect to $\theta$,
$f_{\rm b}$ is as in~(\ref{eq:backupsystem}),
and $I$ is the ${n \times n}$ identity matrix.

The optimization problem~(\ref{eq:backupOP}) contains infinitely many constraints parameterized by ${\theta \in [0,T]}$.
For computational tractability, they are usually discretized into finitely many, $N_{\rm c}$ constraints at
${\theta_{i} = i T/N_{\rm c}}$,
${i \in \mathcal{I} = \{ 0, 1, \ldots, N_{\rm c} \}}$,
yielding:
\begin{align}
\begin{split}
    k(x) \!=\! \underset{u \in \U}
    {\operatorname{argmin}}
    & \ \| u - k_{\rm d}(x) \|_{\Gamma}^2 \\
    \text{s.t.}
    & \ \dot{\bar{h}}_i(x,u) \!\geq\! - \alpha \big( \bar{h}_i(x) \big), \ \forall i \!\in\! \mathcal{I}, \\
    & \ \dot{\bar{h}}_{\rm b} (x,u) \!\geq\! - \alphab \big( \bar{h}_{\rm b}(x) \big).
\end{split}
\label{eq:backupOPdiscretized}
\end{align}
Here, the shorthand notations
${\bar{h}_i(x) = h(\phib(\theta_i,x))}$ and
${\bar{h}_{\rm b}(x) = \hb(\phib(T,x))}$ are used.
In what follows, we implement controller~(\ref{eq:backupOPdiscretized}) in an example.

\begin{example}[Unicycle model] \label{ex:unicycle}
Consider the unicycle model:
\begin{align}
\begin{split}
    \dot{\xi} & = v \cos \psi, \\
    \dot{\eta} & = v \sin \psi, \\
    \dot{\psi} & = \omega,
\end{split}
\label{eq:unicycle}
\end{align}
where the planar position ${p = \begin{bmatrix} \xi & \eta \end{bmatrix}^\top \in \R^2}$ and yaw angle ${\psi \in \R}$ constitute the state ${x = \begin{bmatrix} p^\top & \psi \end{bmatrix}^\top}$, while the speed ${v \in [v_{\min},v_{\max}] \subset \R}$ and yaw rate ${\omega \in [-\omega_{\max},\omega_{\max}] \subset \R}$ form the control input ${u = \begin{bmatrix} v & \omega \end{bmatrix}^\top}$.
We seek to drive the unicycle in the $\xi$ direction at a goal position $\eta_{\rm g}$ with a speed $v_{\rm g}$, while avoiding a circular obstacle of radius ${R_{\rm O}>0}$ at position $p_{\rm O}$.
First we consider a stationary obstacle, then a moving obstacle with velocity $\dot{p}_{\rm O}$ and acceleration $\ddot{p}_{\rm O}$.
Note that this latter problem is well-studied~\cite{Exarchos2015}, and closed-form expressions of control invariant sets exist~\cite{shoukry2017closedform}.

We realize the target motion by the desired controller:
\begin{equation}
    k_{\rm d}(x) = \begin{bmatrix}
    v_{\rm g} \\
    K_\eta (\eta_{\rm g} - \eta) - K_\psi \sin \psi
    \end{bmatrix},
\end{equation}
that is to be modified to obtain a safe controller $k(x)$.
To characterize safety, we first introduce the Eucledian distance $D$ from the obstacle center, the normal vector $n$ pointing away from the obstacle, and a related projection matrix $P$:
\begin{equation}
    D = \| p - p_{\rm O} \|, \quad
    n = \frac{p - p_{\rm O}}{\| p - p_{\rm O} \|}, \quad
    P = I - n n^\top.
\end{equation}
Notice that ${\partial D/\partial p = n^\top}$ and
${\partial n/\partial p = P/D}$ hold.
Furthermore, let us describe the heading direction by:
\begin{equation}
    q = \begin{bmatrix}
    \cos \psi \\ \sin \psi
    \end{bmatrix}, \quad
    r = \begin{bmatrix}
    -\sin \psi \\ \cos \psi
    \end{bmatrix}.
\end{equation}

With these preliminaries, we introduce the following function from~\cite{molnar2022modelfree} to characterize safety:
\begin{equation}
    h(x,t) = D - R_{\rm O} + \delta n^\top q,
\end{equation}
where a tunable parameter ${\delta \geq 0}$ penalizes heading towards the obstacle.
The corresponding derivatives read:
\begin{align}
\begin{split}
    \gradh(x,t) & = \begin{bmatrix}
    n^\top + \delta q^\top P/D & \delta n^\top r
    \end{bmatrix}, \\
    \derp{h}{t}(x,t) & = -n^\top \dot{p}_{\rm O} - \delta q^\top P \dot{p}_{\rm O}/D.
\end{split}
\end{align}
Note that $h$ explicitly depends on time through $p_{\rm O}$ if the obstacle is moving, and one must include $\partial h / \partial t$ in $\dot{h}$.
For stationary obstacle, this dependence on $t$ can be omitted.

Without input bounds, $h$ could be used as CBF and controller~(\ref{eq:QP}) would ensure safe behavior.
The result of executing~(\ref{eq:QP}) while excluding the input bounds (i.e., taking ${\U = \R^2}$) is illustrated in Fig.~\ref{fig:unicycle}(a) for the parameters in Table~\ref{tab:parameters} and ${x_0=0}$.
For ${\delta=0}$, i.e., when the heading direction is not penalized by the CBF, the unicycle stops in front of the obstacle (dashed line), which is safe but overly conservative.
For ${\delta>0}$, the unicycle safely executes the task (solid line).
However, since the input bounds are not incorporated into the optimization problem, the lower and upper speed limits are violated.
On the other hand, $h$ is not necessarily a valid CBF in the presence of input bounds.

To address input bounds, we rely on the backup controller:
\begin{equation}
    k_{\rm b}(x,t) = \begin{bmatrix}
    v_{\max} \\
    \omega_{\max} \tanh(n^\top r / \varepsilon)
    \end{bmatrix},
\end{equation}
that seeks to turn the unicycle away from the obstacle as fast as possible and drive away with maximum speed.
Parameter $\varepsilon$ tunes the aggressiveness of turning, and the yaw rate $\pm \omega_{\max}$ is achieved as ${\varepsilon \to 0}$.
Hence, this controller allows us to keep safety against obstacles that move slower than $v_{\max}$ and turn slower than $\omega_{\max}$.

The backup controller is associated with:
\begin{equation}
    \hb(x,t) = n^\top (q v_{\max} - \dot{p}_{\rm O}),
\end{equation}
whose derivatives are:
\begin{align}
\begin{split}
    \gradhb(x,t) = \begin{bmatrix}
    (q v_{\max} - \dot{p}_{\rm O})^\top P/D & n^\top r v_{\max}
    \end{bmatrix}, \\
    \derp{\hb}{t}(x,t) = -(q v_{\max} - \dot{p}_{\rm O})^\top P \dot{p}_{\rm O}/D - n^\top \ddot{p}_{\rm O}.
\end{split}
\end{align}
We remark that the backup set that is kept invariant by $\kb$ is in fact given by both
${\hb(x,t) \geq 0}$ and
${h(x,t) \geq 0}$,
 (rather than just ${\hb(x,t) \geq 0}$), but both of these functions are involved in~(\ref{eq:backup_condition}).

The efficacy of the backup set method with controller~(\ref{eq:backupOPdiscretized}) is shown in Fig.~\ref{fig:unicycle}(b) for parameters in Table~\ref{tab:parameters} and ${\delta = 0}$.
The controller maintains safety while satisfying the input bounds, and note that even ${\delta = 0}$ yields desired behavior.
The same controller is tested for the case of a moving obstacle in Fig.~\ref{fig:unicycle}(c).
The obstacle moves in the $\eta$ direction sinusoidally, with
${p_{\rm O}(t) = \begin{bmatrix} \xi_{\rm O} & \eta_{\rm O}(t) \end{bmatrix}^\top}$,
${\eta_{\rm O}(t) = \bar{\eta}_{\rm O} - A_\eta \sin(\Omega t) / \Omega}$.
The end result is still safety with bounded inputs.
\end{example}

\bgroup
\setlength{\tabcolsep}{3pt}
\begin{table}
\caption{Parameter Values for the Numerical Examples}
\begin{center}
\begin{tabular}{|c|c|c||c|c|c|}
\hline
Parameter & Value & Unit & Parameter & Value & Unit \\
\hline
$v_{\rm max}$ & 0.2 & m/s & $\delta$ (Ex.~\ref{ex:unicycle}) & 0, 0.5 & m \\
$v_{\rm min}$ & 0.1 & m/s & $\delta$ (Ex.~\ref{ex:quadruped}) & 0 & m \\
$\omega_{\rm max}$ & 0.3 & rad/s & $\varepsilon$ & 0.01 & 1 \\
$\eta_{\rm g}$ & 0 & m & $\gamma$ & 1 & 1/s \\
$v_{\rm g}$ & 0.2 & m/s & $\gammab$ & 1 & 1/s \\
$K_\eta$ & 0.5 & 1/(ms) & $\Gamma$ & ${\rm diag}\{1, 0.25\}$ & \{1, m$^2$\} \\
$K_\psi$ & 0.5 & 1/s & $T$ & 4 & s \\
$R_{\rm O}$ & 0.75 & m & $N_{\rm c}$ (Ex.~\ref{ex:unicycle}) & 80 & 1 \\
$\xi_{\rm O}$ & 2 & m & $N_{\rm c}$ (Ex.~\ref{ex:quadruped}) & 400 & 1 \\
$\bar{\eta}_{\rm O}$ & -0.25 & m & $\sigma$ & 0.1 & m/s \\
$A_\eta$ & $0.1$ & m/s & $\sigmab$ & 0.1 & m \\
$\Omega$ & $2\pi/5$ & rad/s & $p$ & $10^{18}$ & 1 \\
& & & $p_{\rm b}$ & $10^{18}$ & s$^2$ \\
\hline
\end{tabular}
\end{center}
\label{tab:parameters}
\end{table}
\egroup

\section{REDUCED ORDER MODELS}
\label{sec:ROM}

Model~(\ref{eq:system}) is often a simplified representation of a real control system.
The actual dynamics may be more complicated, higher dimensional, involving unmodeled phenomena.
Hence, we call~(\ref{eq:system}) as {\em reduced order model (ROM)}.
The backup set method is able to control the ROM with formal safety guarantees while respecting input constraints.
Yet, the safety of the actual {\em full order system (FOS)} is not necessarily ensured.
Next, we investigate the effect of unmodeled dynamics on safety, and derive conditions for the safety of the FOS by following our previous work~\cite{molnar2022modelfree}.
Then, we propose a robustified backup set method.
We consider the ROM to be given, while approaches to construct ROMs are out of scope of this paper.
Finally, we demonstrate our framework on an example, in which the locomotion of a quadruped (FOS) is controlled to follow the unicycle model (ROM).

Consider a FOS given by state ${X \in \R^N}$, input ${U \in \R^M}$, locally Lipschitz continuous functions
${F : \R^N \to \R^N}$ and
${G : \R^N \to \R^{N \times M}}$, and dynamics:
\begin{equation}
    \dot{X} = F(X) + G(X) U.
    \label{eq:fullsystem}
\end{equation}
Furthermore, let a {\em reduced order state} ${x \in \R^n}$ be defined by a continuously differentiable map ${P : \R^N \to \R^n}$:
\begin{equation}
    x = P(X).
\label{eq:reducedstate}
\end{equation}
The reduced order state is selected such that it describes safety-critical behavior.
Specifically, consider the safe set:
\begin{equation}
    \C = \{X \in \R^N: h(P(X)) \geq 0 \}
    \label{eq:fullsafeset}
\end{equation}
for the FOS with ${h : \R^n \to \R}$ given as before.

To achieve safe FOS behavior, one may construct a ROM like~(\ref{eq:system}), design a safety-critical ROM controller ${u = k(x)}$, and utilize a tracking controller ${K : \R^N \times \R^m \to \R^M}$, ${U = K(X,u)}$ so that the closed-loop FOS:
\begin{equation}
    \dot{X} = F(X) + G(X) K(X,u)
    \label{eq:fullclosedloop}
\end{equation}
tracks the ROM.
With appropriate ROM and tracking controller, the true dynamics of the reduced order state $x$ track the ROM accurately.
The true reduced order dynamics are:
\begin{equation}
    \dot{x} = f(x) + g(x) u + d,
    \label{eq:disturbedsystem}
\end{equation}
where ${d \in \R^n}$ is the deviation from the ROM, given by:
\begin{equation}
    d = \gradP (X) \big( F(X) \!+\! G(X) K(X,u) \big) \!-\! f(P(X)) \!-\! g(P(X)) u.
\end{equation}
Note that while $d$ acts as disturbance on the ROM, it can be viewed as tracking error that the FOS seeks to eliminate.

If the discrepancy $d$ is zero, the ROM captures the safety-critical behavior of the FOS accurately, and the backup set method can be used directly with the control invariant set:
\begin{equation}
    \CI \!=\! \left\{\! X \!\in\! \R^N :
    \!\!\begin{array}{l}
    h \big( \phib(\theta,P(X)) \big) \!\geq\! 0, \ \forall \theta \!\in\! [0,T], \\
    \hb \big( \phib(T,P(X)) \big) \!\geq\! 0 
    \end{array}
    \!\!\right\},
\end{equation}
for which ${X \in \CI \iff x \in \SI}$.
Then, per Theorem~\ref{thm:backup}, there exists a controller $k$ that satisfies~(\ref{eq:backup_condition}) and renders ${\CI \subseteq \C}$ forward invariant along~(\ref{eq:fullsystem}).
However, nonzero discrepancy $d$ may lead to safety violations.
Below we discuss conditions under which safety is preserved, and we investigate how to provide robustness against $d$.
During robustification,~(\ref{eq:disturbedsystem}) is considered while the discrepancy $d$ is viewed as an unknown but bounded term (see assumptions below) that represents modeling errors and disturbances associated with the ROM.

\subsection{Safety with Ideal Tracking}

If the ROM and the tracking controller are well-designed, the true reduced order dynamics converges to the ROM and the discrepancy $d$ vanishes.
First, we consider this ideal scenario as reflected by the following assumption.
\begin{assumption} \label{assum:tracking}
The tracking controller ${U = K(x,u)}$ drives the discrepancy between the true reduced order dynamics and the ROM to zero exponentially.
That is, there exist ${A \geq 0}$ and ${\lambda>0}$ such that ${\forall t \geq 0}$:
\begin{equation}
    \|d\| \leq A {\rm e}^{-\lambda t} .
\end{equation}
\end{assumption}
\noindent For simplicity, we assume exponential convergence, although one could also consider asymptotic stability with a class-$\mathcal{KL}$ function on the right-hand side.
Similarly, to simplify our discussion, we choose linear class-$\Kinf$ functions:
${\alpha(r) = \gamma r}$,
${\alphab(r) = \gammab r}$,
with ${\gamma, \gammab > 0}$.
Furthermore, we assume that the gradients of ${h(\phib(\theta, x))}$ and ${\hb(\phib(\theta, x))}$ are bounded, i.e., there exist ${D, D_{\rm b} \geq 0}$ such that
${\| \partial h \big( \phib(\theta,x) \big) / \partial x \| \leq D}$,
${\forall \theta \in [0,T]}$ and
${\| \partial \hb \big( \phib(\theta,x) \big) / \partial x \| \leq D_{\rm b}}$
hold ${\forall x \in \SI}$
with the Euclidean norm ${\| . \|}$.
These assumptions are relaxed in the next section.

Under these assumptions, we show that a time-varying subset of $\CI$ is control invariant.
We define this set $\Cd(t)$ by:
\begin{equation}
    \Cd(t) = \left\{ X \!\in\! \R^N :
    \begin{array}{l}
    H(\theta,X,t) \geq 0, \ \forall \theta \!\in\! [0,T], \\
    \Hb(T,X,t) \geq 0 
    \end{array}
    \right\},
    \label{eq:invariantset_timevarying}
\end{equation}
with:
\begin{align}
\begin{split}
    H(\theta,X,t) & = h \big( \phib(\theta,P(X)) \big) - \frac{D A {\rm e}^{-\lambda t}}{\lambda - \gamma}, \\
    \Hb(T,X,t) & = \hb \big( \phib(T,P(X)) \big) - \frac{D_{\rm b} A {\rm e}^{-\lambda t}}{\lambda - \gammab}.
\end{split}
\label{eq:timevarying_extension}
\end{align}

\begin{theorem} \label{thm:tracking}
\textit{
Consider the ROM~(\ref{eq:system}), set $\SI$ in~(\ref{eq:invariantset}) and a locally Lipschitz continuous controller ${k : \R^n \to \U}$ that satisfies~(\ref{eq:backup_condition}) with
${\alpha(r) = \gamma r}$,
${\alphab(r) = \gammab r}$,
${\forall x \in \SI}$.
Furthermore, consider the FOS~(\ref{eq:fullsystem}), reduced order state~(\ref{eq:reducedstate}), set $\C$ in~(\ref{eq:fullsafeset}), set $\Cd(t)$ in~(\ref{eq:invariantset_timevarying})-(\ref{eq:timevarying_extension}), and Assumption~\ref{assum:tracking}.
If ${\gamma, \gammab < \lambda}$, then ${\Cd(t) \subseteq \C}$ is forward invariant along~(\ref{eq:fullclosedloop}).}
\end{theorem}

\begin{proof}
For ${\gamma, \gammab < \lambda}$, the terms ${D A / (\lambda - \gamma) {\rm e}^{-\lambda t}}$ and ${D_{\rm b} A / (\lambda - \gammab) {\rm e}^{-\lambda t}}$ in~(\ref{eq:timevarying_extension}) are nonnegative ${\forall t \geq 0}$.
Hence, ${\Cd(t) \subseteq \CI}$, ${\forall t \geq 0}$, and ${X \in \Cd(t)}$ implies ${x \in \SI}$.
Then, Theorem~\ref{thm:backup} can be applied, and a controller $k$ satisfying~(\ref{eq:backup_condition}) is guaranteed to exist for all ${X \in \Cd(t)}$ since ${x \in \SI}$.
Given~(\ref{eq:timevarying_extension}) and~(\ref{eq:backup_condition}), the derivative of $H$ along~(\ref{eq:fullclosedloop}) satisfies:
\begin{align}
\begin{split}
    \dot{H}&(\theta, X, t, k(P(X)), d) \\
    & = \dot{h} \big( \phib(\theta,x), k(x) \big) + \derp{h \big( \phib(\theta, x) \big)}{x} d + \frac{\lambda D A {\rm e}^{-\lambda t}}{\lambda - \gamma} \\
    & \geq - \gamma h(\phib(\theta,x)) - \bigg\| \derp{h \big( \phib(\theta, x) \big)}{x} \bigg\| \| d \| + \frac{\lambda D A {\rm e}^{-\lambda t}}{\lambda - \gamma} \\
    & \geq - \gamma h(\phib(\theta,x)) - D A {\rm e}^{-\lambda t} + \frac{\lambda D A {\rm e}^{-\lambda t}}{\lambda - \gamma} \\
    & \geq - \gamma H(\theta, X, t).
\end{split}
\label{eq:tracking_proof}
\end{align}
Similarly, ${\dot{H}_{\rm b}(T, X, t, k(P(X)), d) \geq - \gammab \Hb(T, X, t)}$ can be proven.
Thus, by Theorem~\ref{thm:CBF} we can conclude the forward invariance of ${\Cd(t) \subseteq \CI}$, that implies a safe FOS.
\end{proof}

\begin{remark}
Theorem~\ref{thm:tracking} states that with fast enough tracking of the ROM, the FOS stays safe and evolves in a region where backup set method-based controllers are guaranteed to exist.
However, this result is conditioned on ideal exponential tracking (and the technical assumption about the bounded gradients of $h$ and $\hb$).
Next, we relax these restrictions.
\end{remark}

\subsection{Input-to-State Safe Backup Set Method}

Let us use the following weaker assumption on tracking.
\begin{assumption} \label{assum:tracking_relaxed}
The tracking controller ${U = K(x,u)}$ drives the discrepancy between the true reduced order dynamics and the ROM to a {\em neighborhood} of zero exponentially.
That is, there exist ${A,B \geq 0}$ and ${\lambda>0}$ such that ${\forall t \geq 0}$:
\begin{equation}
    \|d\|^2 \leq A {\rm e}^{-\lambda t} + B.
    \label{eq:ISS_tracking}
\end{equation}
\end{assumption}
\noindent Note that this assumption includes the case ${A=0}$, i.e., when the discrepancy does not necessarily decay but stays bounded below $B$.
We also remark that the square after the norm of $d$ is introduced for algebraic convenience only.
The assumption is shown to hold for the quadruped example below.

When the discrepancy does not decay to zero (${B \neq 0}$), safety can no longer be formally guaranteed by~(\ref{eq:backup_condition}).
To remedy this, some CBF approaches add extra robustifying terms to their safety constraints~\cite{ames2019issf, jankovic2018robust}.
For example, the approach of {\em input-to-state safe} CBFs modifies~(\ref{eq:safety_condition}) to ${\dot{h} \big( x, k(x) \big) \geq - \alpha \big( h(x) \big) + \sigma \| \gradh(x) \|^2}$ with ${\sigma > 0}$ (where ${\gradh(x)}$ could be replaced with ${\gradh(x) g(x)}$ in case of matched disturbances)~\cite{Alan2022}.
We propose to extend this approach to the {\em input-to-state safe backup set method}, by modifying~(\ref{eq:backup_condition}) to:
\begin{align}
\begin{split}
    \dot{h} \big( \phib(\theta,x), k(x) \big) & \geq - \alpha \big( h(\phib(\theta,x)) \big) \\
    & \!\!\! + \sigma \bigg\| \derp{h(\phib(\theta,x))}{x} \bigg\|^2, \ \forall \theta \in [0,T], \\
    \dot{h}_{\rm b} \big( \phib(T,x), k(x) \big) & \geq - \alphab \big( \hb(\phib(T,x)) \big) \\
    & \quad + \sigmab \bigg\| \derp{\hb(\phib(T,x))}{x} \bigg\|^2,
\end{split}
\label{eq:backup_condition_robust}
\end{align}
with tunable parameters ${\sigma, \sigmab > 0}$.

The approach of input-to-state safe CBFs is able to keep a neighborhood of the safe set invariant even with disturbances, and this neighborhood can be tuned as small as desired by parameter $\sigma$.
We seek to achieve the same results with input constraints using the backup set method.
Accordingly, we consider a neighborhood $\Sd$ of the control invariant set $\SI$:
\begin{equation}
    \Sd \!=\! \left\{\! x \!\in\! \R^n \!\!:
    \!\!\!\begin{array}{l}
    h(\phib(\theta,x)) \!\!\geq\!\! -B/(4 \sigma \gamma), \, \forall \theta \!\in\! [0,T], \\
    \hb(\phib(T,x)) \!\!\geq\!\! -B/(4 \sigmab \gammab)
    \end{array}
    \!\!\!\right\}\!\!,
    \label{eq:invariantset_neighborhood}
\end{equation}
determined by $\sigma$, $\sigmab$, and we redefine set $\C_{\rm d}(t)$ in~(\ref{eq:invariantset_timevarying}) with:
\begin{align}
\begin{split}
    H(\theta,X,t) & \!=\! h \big( \phib(\theta,P(X)) \big) \!-\! \frac{A {\rm e}^{-\lambda t}}{4 \sigma (\lambda - \gamma)} \!+\! \frac{B}{4 \sigma \gamma}, \\
    \Hb(T,X,t) & \!=\! \hb \big( \phib(T,P(X)) \big) \!-\! \frac{A {\rm e}^{-\lambda t}}{4 \sigmab (\lambda - \gammab)} \!+\! \frac{B}{4 \sigmab \gammab}.
\end{split}
\label{eq:timevarying_extension_robust}
\end{align}
Then, we state the invariance of set $\Cd(t)$ that can be made arbitrarily close to the safe set $\C$ by increasing $\sigma$, $\sigmab$.

\begin{figure*}
\centering
\includegraphics[scale=1]{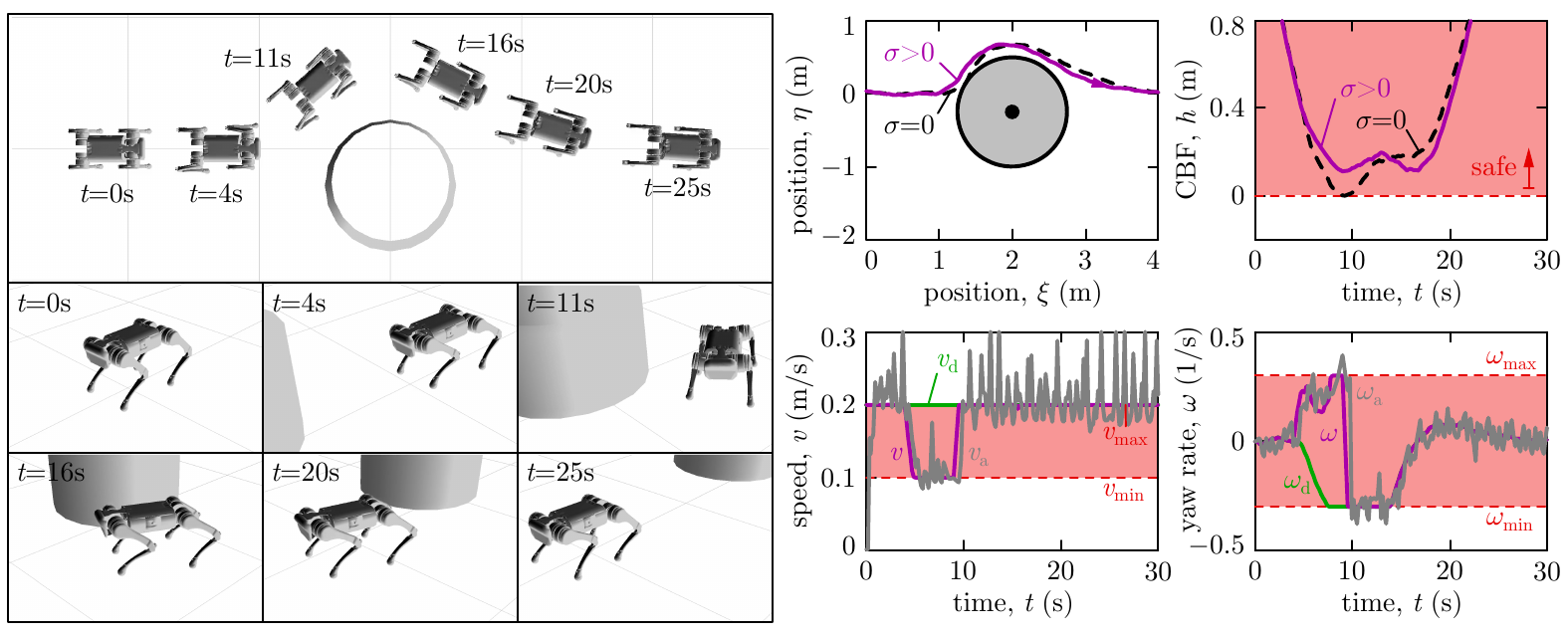}
\caption{
Application of the proposed safety-critical control framework with input constrained reduced order model in quadrupedal locomotion.
The quadruped safely navigates by tracking the speed and yaw rate synthesized with the unicycle model and the input-to-state safe backup set method.
}
\label{fig:quadruped}
\end{figure*}

\begin{theorem} \label{thm:robust}
\textit{
Consider the ROM~(\ref{eq:system}), set $\Sd$ in~(\ref{eq:invariantset_neighborhood}) and a locally Lipschitz continuous controller ${k : \R^n \to \U}$ that satisfies~(\ref{eq:backup_condition_robust}) with
${\alpha(r) = \gamma r}$,
${\alphab(r) = \gammab r}$,
${\forall x \in \Sd}$.
Furthermore, consider the FOS~(\ref{eq:fullsystem}), reduced order state~(\ref{eq:reducedstate}),
set $\Cd(t)$ in~(\ref{eq:invariantset_timevarying})-(\ref{eq:timevarying_extension_robust}), and Assumption~\ref{assum:tracking_relaxed}.
If ${\gamma, \gammab < \lambda}$, then
$\Cd(t)$ is forward invariant along~(\ref{eq:fullclosedloop}).}
\end{theorem}

\begin{proof}
First, we show that the following inequality holds:
\begin{align}
\begin{split}
    \sigma & \bigg\| \derp{h(\phib(\theta,x))}{x} \bigg\|^2 - \bigg\| \derp{h(\phib(\theta,x))}{x} \bigg\| \| d \| \\
    & \geq \bigg( \sqrt{\sigma} \bigg\| \derp{h(\phib(\theta,x))}{x} \bigg\| - \frac{\| d \|}{2 \sqrt{\sigma}} \bigg)^2 - \frac{\| d \|^2}{4 \sigma} \\
    & \geq - \frac{A {\rm e}^{-\lambda t} + B}{4 \sigma}.
\end{split}
\label{eq:ISSf_proof}
\end{align}
Then, the rest of the proof follows that of Theorem~\ref{thm:tracking}:
\begin{align}
\begin{split}
    \dot{H}&(\theta, X, t, k(P(X)), d) \\
    & = \dot{h} \big( \phib(\theta,x), k(x) \big) \!+\! \derp{h \big( \phib(\theta, x) \big)}{x} d \!+\! \frac{\lambda A {\rm e}^{-\lambda t}}{4 \sigma (\lambda \!-\! \gamma)} \\
    & \geq - \gamma h(\phib(\theta,x)) - \frac{A {\rm e}^{-\lambda t} + B}{4 \sigma} + \frac{\lambda A {\rm e}^{-\lambda t}}{4 \sigma (\lambda - \gamma)} \\
    & \geq - \gamma H(\theta, X, t),
\end{split}
\end{align}
cf.~(\ref{eq:tracking_proof}), where in the second step we used the Cauchy-Schwartz inequality and substituted~(\ref{eq:ISSf_proof}).
\end{proof}

\begin{remark}
Theorem~\ref{thm:robust} states that input-to-state stable tracking of the ROM (i.e., when the discrepancy $d$ decays to or is within a neighborhood of zero) makes the FOS stay in the set $\Cd(t)$.
This set can be tuned to be as close to the safe set $\C$ as desired using $\sigma$ and $\sigmab$, cf.~(\ref{eq:timevarying_extension_robust}), and it no longer depends on the bounds of the gradients of $h$ and $\hb$.
However, we cannot claim the existence of a controller $k$ satisfying~(\ref{eq:backup_condition_robust}) ${\forall x \in \Sd}$ anymore.
Hence, instead of~(\ref{eq:backupOPdiscretized}), one may implement a relaxed optimization problem:
\begin{align}
\begin{split}
    k(x) \!=\! \underset{\substack{u \in \U \\ \delta_i, \delta_{\rm b} \geq 0}}
    {\operatorname{argmin}}
    & \ \| u - k_{\rm d}(x) \|_{\Gamma}^2 + \sum_{i \in \mathcal{I}} p_i \delta_i^2 + p_{\rm b} \delta_{\rm b}^2 \\
    \text{s.t.}
    & \ \dot{\bar{h}}_i(x,u) \!\geq\! - \alpha \big( \bar{h}_i(x) \big) \!+\! \sigma \bigg\| \derp{\bar{h}_i(x)}{x} \bigg\|^2 \!\!\!-\! \delta_i, \\
    & \qquad \qquad \qquad \qquad \qquad \qquad \qquad \forall i \!\in\! \mathcal{I}, \\
    & \ \dot{\bar{h}}_{\rm b} (x,u) \!\geq\! - \alphab \big( \bar{h}_{\rm b}(x) \big) \!+\! \sigmab \bigg\| \derp{\bar{h}_{\rm b}(x)}{x} \bigg\|^2 \!\!\!-\! \delta_{\rm b},
\end{split}
\label{eq:backupOProbust}
\end{align}
with
slack variables ${\delta_i, \delta_{\rm b} \geq 0}$ and penalties ${p_i, p_{\rm b} \gg 1}$, ${i \in \mathcal{I}}$.
Formulating provably safe and feasible controllers without this relaxation is subject to future research.
\end{remark}

\begin{example}[Quadrupedal locomotion] \label{ex:quadruped}
Consider the Unitree A1 quadrupedal robot shown in Fig.~\ref{fig:quadruped}.
We seek to execute legged locomotion and accomplish the obstacle avoidance task of Example~\ref{ex:unicycle}.
We consider the quadruped as FOS, and we rely on an existing walking controller for locomotion with given speed and yaw rate.
As such, the walking tracks the unicycle model in Example~\ref{ex:unicycle}, which serves as ROM.

The quadruped has 18 degrees of freedom and 12 actuated joints.
Its motion is described by the configuration ${q \in \R^{18}}$, velocities ${\dot{q} \in \R^{18}}$, inputs ${U \in \R^{12}}$ and holonomic constraints ${c(q) \equiv 0 \in \R^{n_{\rm c}}}$ at the $n_{\rm c}$ number of contacts between the feet and the ground.
The dynamics are governed by the Euler-Lagrange equations:
\begin{align}
\begin{split}
    D(q) \ddot{q} + H(q,\dot{q}) & = B U + J(q)^\top \lambda, \\
    J(q)\ddot{q} + \dot{J}(q,\dot{q})\dot{q} &= 0, \label{eq:EulerLagrange}
\end{split}
\end{align}
with mass matrix ${D(q)\in\mathbb{R}^{18 \times 18}}$, Coriolis and gravity terms ${H(q,\dot{q}) \in \R^{18}}$, Jacobian ${J(q) = \partial c(q)/\partial q \in \mathbb{R}^{n_c \times 18}}$, and constraint wrench ${\lambda \in \mathbb{R}^{n_c}}$.
This yields the FOS~(\ref{eq:fullsystem}) with
the state ${X = \begin{bmatrix} q^\top & \dot{q}^\top \end{bmatrix}^\top \in \R^{36}}$ and expressions:
\begin{equation}
    F(X) \!\!=\!\!
    \begin{bmatrix}
    \dot{q} \\
    -D(q)^{-1} \!\big( H(q,\dot q) \!-\! J(q)^\top \!\lambda \big)\!
    \end{bmatrix}\!\!, \
    G(X) \!\!=\!\!
    \begin{bmatrix}
    0 \\
    D(q)^{-1} B
    \end{bmatrix}\!\!.
\end{equation}

During obstacle avoidance, safety is determined by the planar body position $\xi$ and $\eta$ and the yaw angle $\psi$, leading to the reduced order state ${x \in \R^3}$ of Example~\ref{ex:unicycle}.
These states are elements of the full state $X$.
The corresponding equations in the FOS~(\ref{eq:EulerLagrange}) reduce to the unicycle model~(\ref{eq:unicycle}) if roll and pitch are neglected, thus the unicycle is chosen as ROM.
For legged locomotion, we use the inverse dynamics quadratic program based walking controller, ${U = K(X,u)}$, specified in~\cite{Ubellacker2021}.
This controller is able to track speed and yaw rate commands in the reduced order input ${u \in \R^2}$ as long as they are below $v_{\max}$ and between $\pm \omega_{\max}$, respectively.
We also prescribe the minimum speed $v_{\min}$ so that the quadruped is not allowed to stop.
We use the input-to-state safe backup set method, with details in Example~\ref{ex:unicycle}, to find safe speed and yaw rate commands within these bounds.

Fig.~\ref{fig:quadruped} shows high-fidelity simulations of the quadrupedal locomotion\footnote{See video at: https://youtu.be/h8-x7-4eqWs.}.
The speed and yaw rate are commanded using the proposed controller~(\ref{eq:backupOProbust}), the formulas in Example~\ref{ex:unicycle}, the parameters in Table~\ref{tab:parameters}, and the CVXOPT solver~\cite{andersen2012cvxopt}.
The radius $R_{\rm O}$, that the robot's center should stay outside of, consists of the radius of the obstacle (${0.45\,{\rm m}}$) and the size of the quadruped (${0.3\,{\rm m}}$).
With the proposed controller, the quadruped successfully navigates around the obstacle as shown by the motion tiles. 
Observe that safety is maintained with respect to the specification $h$.
Meanwhile, speed and yaw rate commands stay within desired bounds (while their actual value may exceed the bounds).
The figure also indicates the tracking performance of the walking controller, by comparing the actual speed $v_{\rm a}$ and yaw rate $\omega_{\rm a}$ (extracted from $X$) to the commands $v$ and $\omega$.
Indeed, the discrepancy $d$ between the commanded velocities
${\dot{x} = \begin{bmatrix} v \cos\psi & v \sin\psi & \omega \end{bmatrix}^\top}$
and the corresponding actual values decays and stays bounded, as in~(\ref{eq:ISS_tracking}) in Assumption~\ref{assum:tracking_relaxed}. 
Finally, the trajectory with the standard backup set method, (i.e., controller~(\ref{eq:backupOPdiscretized}) and ${\sigma = 0}$, ${\sigmab = 0}$) is shown by dashed lines.
This case gets closer to safety violations due to lack of robustness to the discrepancy between the ROM and FOS.

This example demonstrates the efficacy of the proposed safety-critical control approach, in which an input constrained ROM is combined with the backup set method and a reliable tracking controller.
The results show safe behavior on a complex quadrupedal robot during obstacle avoidance.
\end{example}

\section{CONCLUSIONS}
\label{sec:concl}

This paper addressed safety-critical control using reduced order models that have bounded inputs.
To formally guarantee safety while respecting input bounds, the backup set method was used.
Robustness with respect to the discrepancy between the reduced order model and the full order system was analyzed.
Conditions were derived for the safety of the full system, and the input-to-state safe backup set method was proposed to robustify against the above mentioned discrepancy.
The efficacy of the proposed control framework was demonstrated by controlling a quadruped for obstacle avoidance while relying on the unicycle model.
Future work includes studying the feasibility of the robustified controller.


\section*{APPENDIX}

\begin{proof}[Proof of Lemma~\ref{lemma:backup_invariance}]
By definition~(\ref{eq:invariantset}) of $\SI$ and Assumption~\ref{assum:backup}, we have:
\begin{equation}
    x \in \SI
    \implies \phib(\theta,x) \in \Sb \subseteq \S, \ \forall \theta \geq T.
\label{eq:SI_invariance_proof_1}
\end{equation}
From this, and the fact that:
\begin{equation}
    \phib(\theta+\vartheta, x_0) = \phib(\theta, \phib(\vartheta,x_0)),
\label{eq:flow_property}
\end{equation}
holds for any arbitrary ${\theta,\vartheta \geq 0}$ and ${x_0 \in \R^n}$, we obtain:
\begin{equation}
    x \in \SI
    \implies \phib(T,\phib(\vartheta,x)) \in \Sb, \ \forall \vartheta \geq 0.
\label{eq:SI_invariance_proof_2}
\end{equation}
Furthermore, the definition~(\ref{eq:invariantset}) of $\SI$ and~(\ref{eq:SI_invariance_proof_1}) give:
\begin{equation}
    x \in \SI
    \implies \phib(\theta,x) \in \S, \ \forall \theta \geq 0.
\label{eq:SI_invariance_proof_3}
\end{equation}
Using the property~(\ref{eq:flow_property}) again, we obtain:
\begin{equation}
    x \!\in\! \SI
    \implies \phib(\theta,\phib(\vartheta,x)) \!\in\! \S, \ \forall \theta \!\in\! [0,T],\ \forall \vartheta \!\geq\! 0.
\label{eq:SI_invariance_proof_4}
\end{equation}
Thus, (\ref{eq:SI_invariance_proof_2}),~(\ref{eq:SI_invariance_proof_4}) and the definition~(\ref{eq:invariantset}) of $\SI$ yield~(\ref{eq:SI_invariance}).
\end{proof}

\begin{proof}[Proof of Lemma~\ref{lemma:backup_feasibility}]
The definition~(\ref{eq:invariantset}) of $\SI$ can be re-written as:
\begin{equation}
    \SI = \left\{ x \in \R^n :
    \begin{array}{l}
    h(\phib(\theta,x)) \geq 0, \ \forall \theta \in [0,T], \\
    \hb(\phib(T,x)) \geq 0 
    \end{array}
    \right\}.
\end{equation}
$\SI$ is rendered forward invariant by the backup controller $\kb$ per Lemma~\ref{lemma:backup_invariance}.
Therefore, Nagumo's theorem~\cite{nagumo1942lage} states:
\begin{align}
\begin{split}
    h(\phib(\theta,x)) = 0 & \implies \dot{h} \big( \phib(\theta,x), \kb(x) \big) \geq 0, \\
    \hb(\phib(T,x)) = 0 & \implies \dot{h}_{\rm b} \big( \phib(T,x), \kb(x) \big) \geq 0.
\end{split}
\label{eq:Nagumo}
\end{align}
Consider the second condition and let:
\begin{equation}
    \tilde{\S}(x) = \{\tilde{x} \in \R^n: \hb(\phib(T,x)) \!\geq\! \hb(\phib(T,\tilde{x})) \!\geq\! 0 \}.
\end{equation}
Note that ${\forall x \in \SI}$, ${\tilde{\S}(x)}$ is nonempty and ${x \in \tilde{\S}(x)}$, thus:
\begin{equation}
    \dot{h}_{\rm b} \big( \phib(T,x), \kb(x) \big) \geq \inf_{\tilde{x} \in \tilde{\S}(x)} \dot{h}_{\rm b} \big( \phib(T,\tilde{x}), \kb(\tilde{x}) \big).
\label{eq:inf_proof_1}
\end{equation}
Now let us define the set $\tilde{\S}_r$ for ${r \geq 0}$ and ${\tilde{\alpha}_{\rm b}: \R_{\geq 0} \to \R}$:
\begin{align}
    \tilde{\S}_r & = \{\tilde{x} \in \R^n: r \geq \hb(\phib(T,\tilde{x})) \geq 0 \}. \\
    \tilde{\alpha}_{\rm b}(r) & = -\inf_{\tilde{x} \in \tilde{\S}_r} \dot{h}_{\rm b} \big( \phib(T,\tilde{x}), \kb(\tilde{x}) \big).
\end{align}
Then,~(\ref{eq:inf_proof_1}) is equivalent to:
\begin{equation}
    \dot{h}_{\rm b} \big( \phib(T,x), \kb(x) \big) \geq - \tilde{\alpha}_{\rm b} \big( \hb(\phib(T,x)) \big).
    \label{eq:inf_proof_2}
\end{equation}
Note that $\tilde{\alpha}_{\rm b}$ is monotonically increasing with respect to $r$ since the $\inf$ is taken over a larger set $\tilde{\S}_r$ as $r$ grows.
Furthermore, $\tilde{\alpha}_{\rm b}$ satisfies ${\tilde{\alpha}_{\rm b}(0) \leq 0}$ based on~(\ref{eq:Nagumo}).
Therefore, there exists ${\alphab \in \Kinf}$ such that ${\alphab(r) \geq \tilde{\alpha}_{\rm b}(r)}$, ${\forall r \geq 0}$.
This, together with~(\ref{eq:inf_proof_2}), leads to the second statement in~(\ref{eq:backup_feasibility}).
The first statement can be proven the same way: showing the existence of ${\alpha_{\theta} \in \Kinf}$  for each $\theta \in [0,T]$ and defining ${\alpha \in \Kinf}$ such that ${\alpha(r) = \max_{\theta \in [0,T]} \alpha_{\theta}(r)}$.
\end{proof}



\vspace{0.3cm}
\noindent \textbf{Acknowledgment.}  
We thank Albert Li and Andrew Taylor for discussions about safety with reduced order models, and Wyatt Ubellacker for his invaluable help in synthesizing low-level controllers for the quadruped.

\bibliographystyle{IEEEtran}
\bibliography{2022_acc}

\begin{thebibliography}{10}
\providecommand{\url}[1]{#1}
\csname url@rmstyle\endcsname
\providecommand{\newblock}{\relax}
\providecommand{\bibinfo}[2]{#2}
\providecommand\BIBentrySTDinterwordspacing{\spaceskip=0pt\relax}
\providecommand\BIBentryALTinterwordstretchfactor{4}
\providecommand\BIBentryALTinterwordspacing{\spaceskip=\fontdimen2\font plus
\BIBentryALTinterwordstretchfactor\fontdimen3\font minus
  \fontdimen4\font\relax}
\providecommand\BIBforeignlanguage[2]{{%
\expandafter\ifx\csname l@#1\endcsname\relax
\typeout{** WARNING: IEEEtran.bst: No hyphenation pattern has been}%
\typeout{** loaded for the language `#1'. Using the pattern for}%
\typeout{** the default language instead.}%
\else
\language=\csname l@#1\endcsname
\fi
#2}}

\bibitem{fawcett2022toward}
R.~T. Fawcett, K.~Afsari, A.~D. Ames, and K.~A. Hamed, ``Toward a data-driven
  template model for quadrupedal locomotion,'' \emph{IEEE Robotics and
  Automation Letters}, vol.~7, no.~3, pp. 7636--7643, 2022.

\bibitem{xiong2022}
X.~Xiong and A.~Ames, ``{3-D} underactuated bipedal walking via {H-LIP} based
  gait synthesis and stepping stabilization,'' \emph{IEEE Transactions on
  Robotics}, vol.~38, no.~4, pp. 2405--2425, 2022.

\bibitem{AmesXuGriTab2017}
A.~D. Ames, X.~Xu, J.~W. Grizzle, and P.~Tabuada, ``Control barrier function
  based quadratic programs for safety critical systems,'' \emph{IEEE
  Transactions on Automatic Control}, vol.~62, no.~8, pp. 3861--3876, 2017.

\bibitem{glotfelter2017nonsmooth}
P.~Glotfelter, J.~Cort{\'{e}}s, and M.~Egerstedt, ``Nonsmooth barrier functions
  with applications to multi-robot systems,'' \emph{IEEE Control Systems
  Letters}, vol.~1, no.~2, pp. 310--315, 2017.

\bibitem{dunlap2022comparing}
K.~Dunlap, M.~Hibbard, M.~Mote, and K.~Hobbs, ``Comparing run time assurance
  approaches for safe spacecraft docking,'' \emph{IEEE Control Systems
  Letters}, vol.~6, pp. 1849--1854, 2022.

\bibitem{gurriet2020scalable}
T.~{Gurriet}, M.~{Mote}, A.~{Singletary}, P.~{Nilsson}, E.~{Feron}, and A.~D.
  {Ames}, ``A scalable safety critical control framework for nonlinear
  systems,'' \emph{IEEE Access}, vol.~8, pp. 187\,249--187\,275, 2020.

\bibitem{agrawal2021safe}
D.~R. Agrawal and D.~Panagou, ``Safe control synthesis via input constrained
  control barrier functions,'' in \emph{60th IEEE Conference on Decision and
  Control}, 2021, pp. 6113--6118.

\bibitem{liu2022safe}
S.~Liu, J.~Dolan, and C.~Liu, ``Safe control under input saturation with neural
  control barrier functions,'' in \emph{6th Annual Conference on Robot
  Learning}, 2022.

\bibitem{ames2019issf}
S.~{Kolathaya} and A.~D. {Ames}, ``Input-to-state safety with control barrier
  functions,'' \emph{IEEE Control Systems Letters}, vol.~3, no.~1, pp.
  108--113, 2019.

\bibitem{Alan2022}
A.~Alan, A.~J. Taylor, C.~R. He, G.~Orosz, and A.~D. Ames, ``Safe controller
  synthesis with tunable input-to-state safe control barrier functions,''
  \emph{IEEE Control Systems Letters}, vol.~6, pp. 908--913, 2022.

\bibitem{jankovic2018robust}
M.~Jankovic, ``Robust control barrier functions for constrained stabilization
  of nonlinear systems,'' \emph{Automatica}, vol.~96, pp. 359--367, 2018.

\bibitem{ding2011reachability}
J.~Ding, E.~Li, H.~Huang, and C.~J. Tomlin, ``Reachability-based synthesis of
  feedback policies for motion planning under bounded disturbances,'' in
  \emph{IEEE International Conference on Robotics and Automation}, 2011, pp.
  2160--2165.

\bibitem{bansal2017HJreachability}
S.~Bansal, M.~Chen, S.~Herbert, and C.~J. Tomlin, ``{Hamilton}-{Jacobi}
  reachability: {A} brief overview and recent advances,'' in \emph{56th IEEE
  Conference on Decision and Control}, 2017, pp. 2242--2253.

\bibitem{Kousik2020}
S.~Kousik, S.~Vaskov, F.~Bu, M.~Johnson-Roberson, and R.~Vasudevan, ``Bridging
  the gap between safety and real-time performance in receding-horizon
  trajectory design for mobile robots,'' \emph{The International Journal of
  Robotics Research}, vol.~39, no.~12, pp. 1419--1469, 2020.

\bibitem{Choi2021}
J.~J. Choi, D.~Lee, K.~Sreenath, C.~J. Tomlin, and S.~L. Herbert, ``Robust
  control barrier–value functions for safety-critical control,'' in
  \emph{60th IEEE Conference on Decision and Control}, 2021, pp. 6814--6821.

\bibitem{chen2021backup}
Y.~Chen, M.~Jankovic, M.~Santillo, and A.~D. Ames, ``Backup control barrier
  functions: Formulation and comparative study,'' in \emph{60th IEEE Conference
  on Decision and Control}, 2021, pp. 6835--6841.

\bibitem{Exarchos2015}
I.~Exarchos, P.~Tsiotras, and M.~Pachter, ``On the suicidal pedestrian
  differential game,'' \emph{Dynamic Games and Applications}, vol.~5, no.~3,
  pp. 297--317, 2015.

\bibitem{shoukry2017closedform}
Y.~Shoukry, P.~Tabuada, S.~Tsuei, M.~B. Milam, J.~W. Grizzle, and A.~D. Ames,
  ``Closed-form controlled invariant sets for pedestrian avoidance,'' in
  \emph{American Control Conference}, 2017, pp. 1622--1628.

\bibitem{molnar2022modelfree}
T.~G. Molnar, R.~K.~Cosner, A.~W.~Singletary, W.~Ubellacker, and A.~D.~Ames,
  ``Model-free safety-critical control for robotic systems,'' \emph{IEEE
  Robotics and Automation Letters}, vol.~7, no.~2, pp. 944--951, 2022.

\bibitem{Ubellacker2021}
W.~Ubellacker, N.~Csomay-Shanklin, T.~G. Molnar, and A.~D. Ames, ``Verifying
  safe transitions between dynamic motion primitives on legged robots,'' in
  \emph{IEEE/RSJ International Conference on Intelligent Robots and Systems},
  2021, pp. 8477--8484.

\bibitem{andersen2012cvxopt}
M.~S. Andersen, J.~Dahl, and L.~Vandenberghe, ``{CVXOPT}: {A} {Python} package
  for convex optimization,'' 2012, available: https://cvxopt.org/.

\bibitem{nagumo1942lage}
M.~Nagumo, ``{\"U}ber die lage der integralkurven gew{\"o}hnlicher
  differentialgleichungen,'' \emph{Proceedings of the Physico-Mathematical
  Society of Japan. 3rd Series}, vol.~24, pp. 551--559, 1942.

\end{thebibliography}

\end{document}